\documentclass[conference]{IEEEtran}
\IEEEoverridecommandlockouts
\usepackage{cite}
\usepackage{amsmath,amssymb,amsfonts}
\usepackage{mathtools}
\usepackage{algorithmic}
\usepackage{graphicx} 
\usepackage{textcomp}
\usepackage{tikz}
\usetikzlibrary{automata, positioning}
\usepackage{balance}
\usepackage{stfloats}
\usepackage{booktabs}
\usepackage{multirow, multicol}
\usepackage{breqn}
\usepackage{dsfont}
\usepackage{amsthm}
\newtheorem{theorem}{Theorem}
\definecolor{Thistle}{rgb}{0.85, 0.75, 0.85}
\definecolor{Dandelion}{rgb}{0.94, 0.88, 0.19}
\definecolor{CadetBlue}{rgb}{0.37, 0.62, 0.63}
\definecolor{Aquamarine}{rgb}{0.5, 1.0, 0.83}

\usepackage{lipsum} % Solo per testo fittizio

\newcommand\copyrighttext{%
  \footnotesize \textcopyright~2025 IEEE. Personal use of this material is permitted.
  Permission from IEEE must be obtained for all other uses, in any current or future media,
  including reprinting/republishing this material for advertising or promotional purposes,
  creating new collective works, for resale or redistribution to servers or lists,
  or reuse of any copyrighted component of this work in other works.
}

\newcommand\copyrightnotice{%
  \begin{tikzpicture}[remember picture,overlay]
    \node[anchor=south,yshift=10pt] at (current page.south) {%
      \fbox{%
        \parbox{\dimexpr\textwidth-2\fboxsep-2\fboxrule\relax}{%
          \copyrighttext
        }%
      }%
    };
  \end{tikzpicture}%
}

\def\BibTeX{{\rm B\kern-.05em{\sc i\kern-.025em b}\kern-.08em
    T\kern-.1667em\lower.7ex\hbox{E}\kern-.125emX}}

\begin{document}

\title{
%Bidirectional Age of Incorrect Information: a New Performance Metric for Status Updates in Virtual Dynamic Environments
Bidirectional Age of Incorrect Information: \\A Performance Metric for Status Updates in\\ Virtual Dynamic Environments
\thanks{This work has been supported by the Italian PRIN project 2022PNRR ``DIGIT4CIRCLE,'' project code P2022788KK.}
}

%\author{Anonymous METACOM submission}

\author{\IEEEauthorblockN{Chiara Schiavo, Manuele Favero, Alessandro Buratto and Leonardo Badia}
\IEEEauthorblockA{Dept. of Information Engineering, University of Padova, via Gradenigo 6/b, 35131 Padua, Italy\\}
Emails: \small{\texttt{\{schiavochi, faveromanu, burattoale, badia\}@dei.unipd.it}}
}

%\author{\IEEEauthorblockN{Chiara Schiavo}
%\and
%\IEEEauthorblockN{Manuele Favero}
%\and
%\IEEEauthorblockN{Alessandro Buratto}
%\and
%\IEEEauthorblockN{Leonardo Badia}

\maketitle

\begin{abstract}
Virtual dynamic environments (VDEs) such as the Metaverse and digital twins (DTs) require proper representation of the interacting entities to map their characteristics within the simulated or augmented space. Keeping these representations accurate and up-to-date is crucial for seamless interaction and system reliability. 
In this paper, we propose bidirectional age of incorrect information (BAoII) to address this aspect. BAoII quantifies the time-dependent penalty paid by an entity in a VDE due to incorrect or outdated knowledge about itself and the overall dynamically changing space. This extends the concept of age of incorrect information for a bidirectional information exchange, capturing that a VDE requires mutual awareness of the entity's own representation, measured in the virtual space, and what the other entities share about their representations. 
Using a continuous-time Markov chain model, we derive a closed-form expression for long-term BAoII and identify a transmission cost threshold for optimal update strategies.
We describe a trade-off between communication cost and information freshness and validate our model through numerical simulations,
demonstrating the impact of BAoII on evaluating system performance and highlighting its relevance for real-time collaboration in the Metaverse and DTs. 
\end{abstract}


\copyrightnotice

\section{Introduction}
The Metaverse is an advanced Virtual Dynamic Environment (VDE) that merges digital and physical entities through the convergence of extended reality (XR), web technologies, and the Internet.
This has emerged as a key research trend, driven by continuous advances in virtual/augmented reality (VR/AR) and digital twins (DTs). VR/AR plays a fundamental role in enabling long-duration immersive experiences, while DTs serve as a crucial step toward the realization of persistent and self-sustaining alternative realities~\cite{ali2023metaverse}. %Currently, research focuses mainly on the development, optimization, and integration of DTs into VR environments to enable realistic and interactive experiences. The next phase of Metaverse development will involve generating native content based on these virtual replicas of physical objects. This progression aims to achieve an autonomous digital counterpart of the real world, where virtual environments evolve independently while maintaining consistency with their physical counterparts~\cite{wang2022survey}.

%\begin{figure}[t!]
%    \centering
%    \includegraphics[width=\linewidth]{img/abstract.pdf}
%    \caption{BAoII application scenarios: Metaverse and Digital Twins. %In the Metaverse scenario, users measure their state within the virtual environment and communicate it to others to perform cooperative actions. In the DT scenario, a physical sensor transmits its state to its virtual replica, which then provides updates and corrections based on a ML system's analysis.
%    }
%    \label{fig:graphicalabstract}
%\end{figure}

The Metaverse depends on strict quality of service (QoS) requirements. For instance, VR video streaming demands bit rates up to 1 Gbps, while motion-to-photon (MTP) latency must remain below 20 ms to ensure smooth interactions. Furthermore, XR applications require real-time control of both virtual and physical objects, which prompts for haptic feedback with end-to-end latency under 1 ms and 99.999\% reliability~\cite{lincoln2016motion}. Meeting these demands on heterogeneous networks calls for a transformative architecture with significant computational and communication capabilities. From the user’s perspective, low latency is essential not only for immersive and continuous experiences, but also to prevent sensory conflict, a major cause of motion sickness~\cite{kundu2021study}. Predictive algorithms can help reduce perceived latency~\cite{wilson2020task, guo2020exploring, schiavo2025sagesemanticdrivenadaptivegaussian}. For DTs, low-latency communication ensures real-time synchronization between physical and virtual entities, enabling accurate decision making and seamless interaction in the Metaverse~\cite{van2022edge}. 

%Achieving these high-performance requirements across heterogeneous networks requires a transformative Metaverse architecture with substantial computational and communication capacity.
%These requirements are, from the user's perspective, necessary to ensure high quality of experience (QoE) through constant updates for a smooth and continuous experience~\cite{10293194}. In the Metaverse, latency, i.e., the delay between a user's motion and the corresponding visual update in the virtual environment, is a key technical factor.
%When low, it provides an immersive experience, whereas high latency not only worsens the user's perception but also causes a mismatch between visual and vestibular sensory inputs, leading to sensory conflict, a major cause of motion sickness~\cite{kundu2021study}. User comfort can be improved by reducing latency through predictive algorithms~\cite{ wilson2016effect, guo2020exploring}.
%Low latency is also crucial for DTs as it ensures real-time interaction between physical and virtual worlds, enabling accurate decision making and seamless immersive experiences in the Metaverse~\cite{van2022edge}.

To measure the freshness of information in a dynamic scenario, we can leverage age of information (AoI)~\cite{9380899}, which quantifies the time elapsed since the generation of the last received data packet. A low AoI value ensures that the information is up-to-date to support precise control and immersive experiences in real time. Optimizing AoI is essential to reduce lags and improve the reliability of short packet communications, enabling efficient and synchronized interactions between the physical and virtual worlds.
%qui introdurrei espansione AoI in AoII e varianti varie
Originally introduced to assess network performance, AoI is applied in various domains, including the Internet of things~\cite{8930830}, autonomous vehicles~\cite{5984917}, and remote control systems~\cite{hatami2021aoi}. 

Over time, variations of AoI have been proposed to address specific application needs. 
One notable extension is age of incorrect information (AoII)~\cite{maatouk2020age}, which considers not only the timeliness of updates but also their accuracy. Unlike AoI, which increases regardless of the accuracy of the information, AoII penalizes only incorrect information over time, making it more suitable for remote real-time estimation scenarios \cite{bonagura2025strategic}.

The use of age-related metrics in the Metaverse has not been well investigated, with only a few studies~\cite{cao2023toward, favero2024strategic, xiao2023adaptive}, despite its potential usefulness in ensuring the QoE and QoS required for the full exploitation of VDEs. 
%However, AoI has several limitations, and we seek a metric that is more suitable for real-time communications. 
%In VDEs, cooperation it's the core of the experience and therefore requires special attention. 
The challenge is that the Metaverse involves a bidirectional exchange of information (such as positions and actions), so AoII alone may not be sufficient. To address this limitation, we propose a novel metric for a pair of agents, which quantifies the time during which either agent does not possess up-to-date information. Our penalty grows larger with the time required to regain complete and accurate knowledge about both agents.

This extension of AoII that includes information timeliness and accuracy in a Metaverse context is called the bidirectional age of incorrect information (BAoII). It not only accounts for an entity's update process, but also quantifies how frequently an entity maintains accurate and up-to-date information about both its own state and that of the other entity. 

This work makes the following key contributions. First, we show how BAoII can be suitable for different VDEs with varying constraints.
Moreover, we characterize BAoII from a mathematical perspective, using a continuous-time Markov chain (CTMC), and by providing examples of its application.  %See Fig.~\ref{fig:penalty_functions} for a visual comparison of AoI, AoII, and BAoII.
Finally, with the goal of minimizing BAoII to ensure continuously updated communication, we evaluate how measurement and transmission costs influence the decision-making process of the involved entities. Specifically, we analyze a memoryless scenario in which an entity, upon measuring its state, immediately decides whether to also transmit to the other entity or not. %Transmissions can only occur concurrently with measurements and cannot be deferred to a later time.
%Finally, w

The remainder of this paper is organized as follows. Sec.~\ref{sec:applications} explores Metaverse applications in which bidirectional communication is crucial. Sec.~\ref{sec:proposedmetric} introduces BAoII and outlines its relevance in VDEs. Sec.~\ref{sec:systemoverview} presents the mathematical formulation and derivation of this metric. Sec.~\ref{sec:results} discusses the simulation results. Finally,  we conclude in Sec.~\ref{sec:conclusions}.
%Sec.~\ref{sec:applications} presents real world applications in which we can leverage BAoII to optimize the transmission system. 

\section{Bidirectional applications}\label{sec:applications}
The Metaverse encompasses a wide range of immersive applications, demanding specific technical parameters to ensure seamless user experiences. Key factors include frame rate, latency, and data transmission rates, which vary significantly depending on the interactivity and immersion of the application \cite{10007756}. 
%High-precision tasks, such as medical simulations, require ultra-low latency and high refresh rates for safety and responsiveness. In contrast, smart cities and digital twins prioritize real-time data synchronization and collaborative decision-making, often tolerating moderate latency while demanding consistent bidirectional communication. 
We outline the technical requirements for three distinct Metaverse applications, highlighting the central role that bidirectional communications plays in each of them. %: (A) high-immersion medical and training simulations, (B) social VR and multiplayer gaming, and (C) smart cities and digital twins.

\subsection{High-Immersion Applications}
Medical and surgical applications in the Metaverse demand extremely high performance to ensure precision and safety. These applications include remote robotic-assisted surgeries, real-time medical training, and immersive patient diagnostics. Any perceptible lag in these environments can compromise the effectiveness of training or, in extreme cases, jeopardize patient safety.
A frame rate of $90$--$144$ fps is necessary to provide smooth visual feedback \cite{10049694}. MTP latency must remain below $7$--$15$ ms to avoid disruptive delays during real-time interactions \cite{dohler2025crucial}. 
Tracking data for hand movements and surgical tools typically ranges between $50$--$300$ bytes per update, transmitted at high frequencies ($\geq$$1.000$ Hz) to ensure accuracy \cite{8933555}. 
Haptic feedback, essential for realistic simulations of soft-tissue interactions or force-sensitive procedures, requires larger packets to convey force and texture details ($200$--$1000$ bytes) \cite{9805573}. For high-resolution medical imaging or augmented reality overlays, high-bitrate compressed video can demand $50$--$400$ Mbps to maintain diagnostic precision \cite{queisner2024surgical}.
% adilkhanov2022haptic, 9674010

In collaborative scenarios, bidirectional communication is essential to synchronize actions and maintain shared awareness. Evaluating latency and feedback consistency between users ensures alignment under strict performance constraints.

\subsection{Social VR and Multiplayer Gaming}\label{subsec:VR}
Social interactions in the Metaverse rely on consistent responsiveness and synchronized avatar movements to create a natural and engaging environment. Applications include virtual meetups, collaborative spaces, and multiplayer games. %, all of which require robust networking infrastructure to handle simultaneous interactions among multiple users.
%Frame rates of $72-120$ fps ensures fluid interactions, reducing motion sickness and improving presence. Latency must be kept below $30-50$ ms to maintain lip-sync accuracy and real-time gesture synchronization \cite{van2024impact, cheng2022we}. 
Frame rates of $72$–-$120$ fps and latencies below $30$-–$50$ ms support lip-sync and gesture accuracy \cite{van2024impact, cheng2022we}.
The avatar tracking data are relatively lightweight ($100$-–$500$ bytes per update) but are transmitted frequently ($30$-–$90$ Hz) to ensure precise movements.
Spatial audio, a key component of immersion in social VR, consumes $64$--$256$ kbps per user. The bandwidth per participant ranges from $10$ to $100$ Mbps, depending on the complexity of the scene and the number of users \cite{plopski2022eye}. 

Efficient peer-to-peer communication protocols can significantly reduce latency and bandwidth demands, improving scalability for large-scale virtual events, while bidirectional user communication ensures synchronized dialogues, joint actions, and a heightened sense of presence.

%In social VR and multiplayer gaming, bidirectional communication between users underpins natural dialogue, joint actions, and mutual presence. A bidirectional evaluation metric enables accurate measurement of synchronization quality and responsiveness in peer-to-peer interactions, directly affecting immersion and engagement.
%Efficient peer‑to‑peer protocols lower latency and bandwidth needs \textcolor{red}{manca un pezzo}, boosting scalability for large‑scale virtual events, while bidirectional user communication ensures synchronized dialogues, joint actions, and a heightened sense of presence in social VR and multiplayer gaming.
%tang2022roadmap

\subsection{Smart Cities and Digital Twins}
Smart cities in the Metaverse rely on real-time DTs, i.e., virtual representations of urban systems \cite{allam2022metaverse}. These twins are used to monitor infrastructure, simulate emergency responses, optimize logistics, and manage energy consumption \cite{favero2024logistics}. Their effectiveness depends on continuous bidirectional data exchange between edge devices, centralized systems, and user interfaces, ensuring aligned decision-making and shared real-time context.
Sensor data is typically transmitted at frequencies ranging from $1$ to $100$ Hz, with packet sizes between $200$ bytes (telemetry) and several Mbytes (video surveillance). Most DTs must react in near real time to physical-world changes, requiring latencies below $50$ ms for general updates and $<20$ ms for time-critical decisions (e.g., autonomous vehicle rerouting). Frame rates of $60$–$90$ fps are sufficient for interactive urban visualization in VR/AR environments. Depending on the complexity of the scene, the bandwidth requirements can range from $10$ to $100$ Mbps.

%When DTs are accessed collaboratively, bidirectional user communication ensures aligned decision-making and shared real-time context. Evaluating bidirectional performance is crucial to guarantee temporal coherence and data consistency across all interacting participants.

%\cite{kusuma2022metaverse}

\section{Proposed Metric: BAoII}\label{sec:proposedmetric}
The concept of AoII, introduced in~\cite{maatouk2020age}, quantifies the time elapsed since the knowledge of an entity has become incorrect due to a drift, that is, a state change or a delayed update.  When a new measurement or update restores the correct information, AoII resets to zero; otherwise, it continues to grow.

Inspired by this approach, we observe that in VDEs (e.g.: Metaverse, Social VR) drifts can occur when information about oneself or about another entity becomes outdated or incorrect. Motivated by these scenarios, where bidirectional interaction is required, we introduce BAoII, which captures the mutual freshness of knowledge between entities. %A comparison of the age-related metrics is shown in Fig.~\ref{fig:penalty_functions}.

%aggiungere intro su entity
We consider a system composed of two entities, such as sensors, DTs, or network nodes, that interact bidirectionally. Each entity can perform a self-measurement to monitor its own state and transmit this information to the other.  For BAoII to reset for an entity, it must have two key pieces of information: (i) knowledge of its own condition and (ii) awareness of the condition of the other entity. Information is treated as binary, either correct or incorrect. Any change in the condition of an entity, such as movement, state transitions, or interactions, causes the information to drift, making its state incorrect until observed. To maintain accurate self-awareness, an entity must continuously monitor and update its own state. If no measurement occurs, outdated information may lead to wrong system control.
%In the absence of such measurements, outdated information may persist, potentially leading to system errors.

However, knowledge of the other entity depends on receiving updated information through communication. Consequently, the global information state of the system is influenced by both self-measurements and transmissions between entities. 
%Thus, the system’s global information state results from both local measurements and inter-entity transmissions. 
The goal is to quantify global knowledge of the system, ensuring synchronization and consistency for all participants.

%\textcolor{red}{ Questo pezzo e' ripetitivo e io lo toglierei semplicemente. non capisco a cosa serva.}
%\textcolor{blue}{In the Metaverse, where real-time interaction and synchronization between avatars, DTs, and AI agents are fundamental~\cite{lam2022human}, BAoII provides a robust framework for evaluating information consistency and reliability in VDEs.}

We focus on a VDE scenario involving two interacting entities. %(e.g.: head-mounted display or sensors). 
Each entity is associated with a process that evolves over time: $X_1(t)$ for entity 1 and $X_2(t)$ for entity 2. These processes are not directly accessible to the entities themselves.
To gain knowledge about their own states, each entity performs a self-observation (measurement) process. We denote these self-sampled estimates as $\widehat{X}_1(t)$ and $\widehat{X}_2(t)$, respectively.

Entities can also share their knowledge with each other. Specifically, entity 1 can acquire knowledge about the other by receiving the information that entity 2 transmits about itself through the process $\widetilde{X}_1(t)$; symmetrically, entity 2 can know entity 1 through $\widetilde{X}_2(t)$.

We consider a discrete time axis, normalized to the time slot duration. 
We define two binary processes $\widehat{Y}_i(t)$ and $\widetilde{Y}_i(t)$ to reflect the correctness of the knowledge held by each entity $i \in \{1, 2\}$ at time $t$:
\begin{IEEEeqnarray}{rCll}
    \widehat{Y}_i(t) &=& X_i(t) \oplus \widehat{X}_i(t) \qquad & \text{(self-knowledge)} \nonumber \\
    \widetilde{Y}_i(t) &=& X_{3{-}i}(t) \oplus \widetilde{X}_i(t) \qquad & \text{(knowledge of the other).} \nonumber 
\end{IEEEeqnarray}

Here, $\oplus$ denotes the bitwise XOR operation, and index $3-i$ denotes the other entity (it is equal to $2$ when $i{=}1$ and vice versa). All these processes take values in $\{\texttt{0}, \texttt{1}\}$, denoting correct knowledge (i.e., no mismatch with the actual state), and error, respectively.
%where \texttt{0}, denotes correct knowledge (i.e., no mismatch with the actual state), and \texttt{1} denotes an error.

Each entity decides when to inform the other about its sampling process by adopting a transmission policy that aims to minimize the average of a particular penalty function.
For a single entity, the objective is to achieve complete knowledge by accurately knowing its own state, the state of the other entity, and ensuring that the other entity also has correct information about its state.
%To capture the global knowledge condition from the perspective of entity 1, we define a new binary variable \( Y_1(t) \in \{0, 1\} \), which indicates whether all relevant knowledge for entity 1 is correct at time \( t \). Specifically, we define
%\begin{equation}\label{eq3}
%Y_1(t) = 
%\begin{cases}
%\texttt{0} & \text{if } \; \widehat{Y}_1(t) = \texttt{0} \; \land \; \widehat{Y}_2(t) = \texttt{0} \; \land \; \widetilde{Y}_1(t) = \texttt{0} \\
%\texttt{1} & \text{otherwise}
%\end{cases}
%\end{equation}

%that states that an entity’s knowledge of itself must match the actual process as closely as possible. The knowledge of the considered entity on the other has to be the most similar possible to the process of measure that the other entity performs, i.e., most of the measures done are transmitted to the other entity to keep update it.
An entity's knowledge of the other must reflect the latest measurement the other entity has performed and shared. 
In fact, an entity cannot directly measure the condition of another; it relies on the other entity to collect and transmit the necessary information.
An entity does not gain a direct benefit from transmitting its measurement information to the other. From its perspective, system knowledge is achieved when it accurately knows its own state through self-measurement and the state of the other entity through the measurement and transmission of the latter. Since transmission does not add information to the sender, it is not required for its BAoII to reset. However, global system knowledge is only achieved when both entities measure and transmit, ensuring full synchronization.

The system information state is defined as a sequence of these four binary values: $\widehat{Y}_1(t)$, $\widetilde{Y}_{1}(t)$, $\widehat{Y}_2(t)$, $\widetilde{Y}_{2}(t)$ $\forall t$, hereafter represented as a sequence of four bits inside square brackets (e.g.: $[\texttt{1},\texttt{0},\texttt{1},\texttt{1}]$). %.(Fig.\ref{fig:knowledge}). 
%Each of these values can be \texttt{0} if the information is correct, \texttt{1} if the information contains errors.
If we consider entity $1$ and one of the first three information values is not correctly updated, a penalty is paid:
\begin{equation}
    \Delta_{err}^{(1)}(t) = \big( \widehat{Y}_1(t) \lor \widehat{Y}_2(t) \lor \widetilde{Y}_1(t) \big) \cdot f^{(1)}(t)
\end{equation}
where $f^{(1)}(t)$ is an increasing time penalty function, paid for not knowing the correct status of the process for a certain period of time. To connect with AoI and AoII \cite{5984917,maatouk2020age}, this penalty increases linearly over time as long as the entity remains in an erroneous state:
\begin{equation}
f^{(1)}(t) = t - \sigma^{(1)}(t)
\end{equation}
\begin{equation*}
\sigma^{(1)}(t) = \max \left\{ t' \leq t \;:\; \widehat{Y}_1(t') = \widehat{Y}_2(t') =  \widetilde{Y}_1(t') = 0 \right\}.
\end{equation*}

Since we want to measure the time between the occurrence of an error and the return of an entity to a correct knowledge condition, what we want to study is the time elapsed until the return to $[\texttt{0},\texttt{0},\texttt{0},\texttt{0}]$ or $[\texttt{0},\texttt{0},\texttt{0},\texttt{1}]$ (if we are considering entity $1$). 
Fig.~\ref{fig:knowledge} graphically represents the different knowledge states of the two entities and the process leading to complete correct bidirectional knowledge, where both entities have accurate information about themselves and each other.
\begin{figure}[t!]
    \centering
    \includegraphics[width=0.8\linewidth]{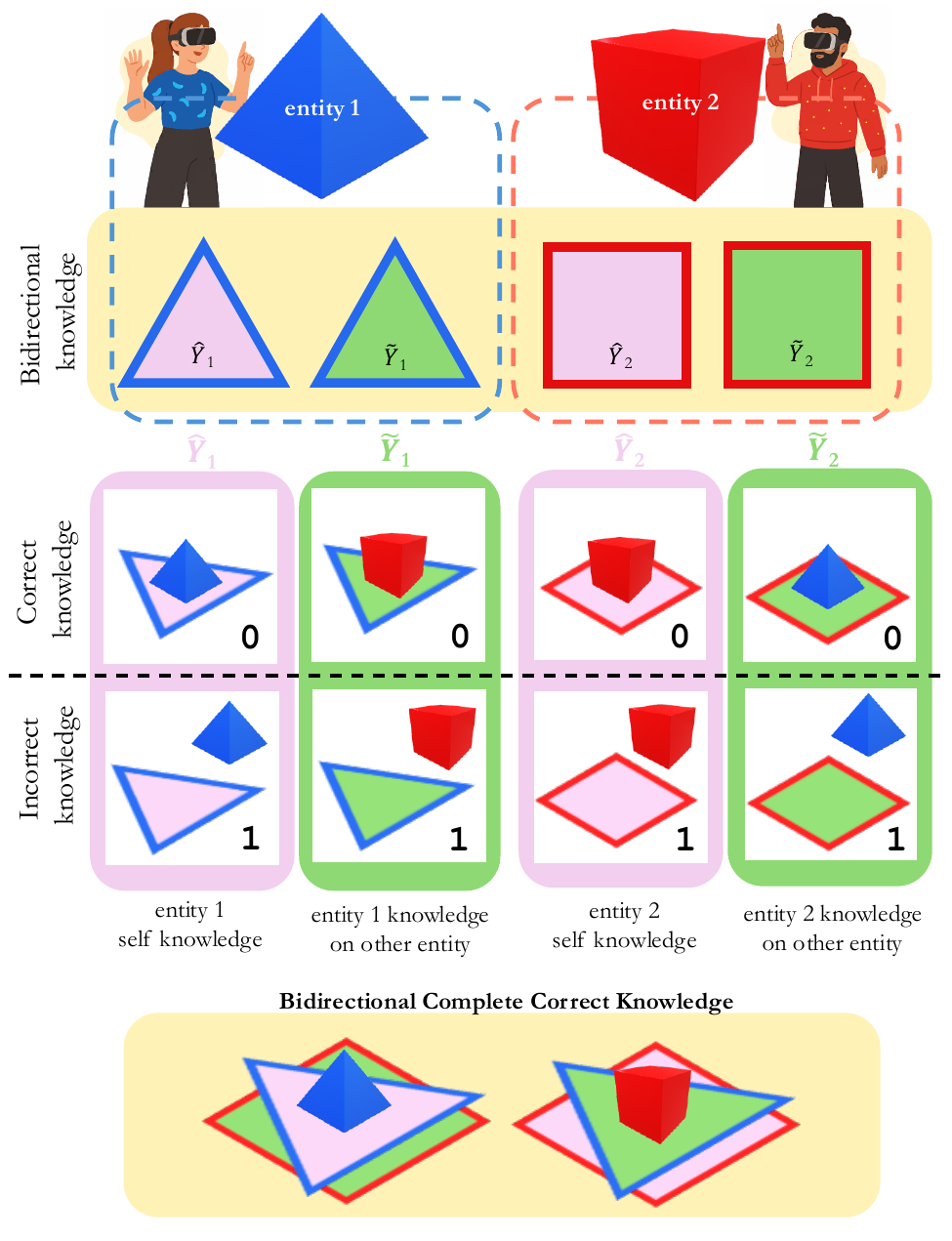}
    \caption{Bidirectional knowledge between two entities.}
    \label{fig:knowledge}
\end{figure}

\begin{table*}[ht]
    \caption{CTMC States Descriptions}
    \centering
    \begin{tabular}{c l l l} % Aggiunta colonna per "Entity 2 State"
        \toprule
        \textbf{State name} & \textbf{State coding} & \textbf{Entity 1 State Description} & \textbf{Entity 2 State Description} \\
        \midrule
        O & $[\texttt{0},\texttt{0},\texttt{0},\texttt{0}]$ & Complete correct information state & Complete correct information state\\
        $\Phi$ & $[\texttt{0},\texttt{0},\texttt{0},\texttt{1}]$ & Complete correct information state & Correct on own condition. Error on Entity 1 condition \\
        A & $[\texttt{1},\texttt{0},\texttt{0},\texttt{1}]$  & Error on own condition. Correct on Entity 2 condition & Correct on own condition. Error on Entity 1 condition \\
        B & $[\texttt{0},\texttt{1},\texttt{1},\texttt{0}]$ & Correct on own condition. Error on Entity 2 condition & Correct on own condition. Error on Entity 1 condition \\
        $\Gamma$ & $[\texttt{0},\texttt{1},\texttt{1},\texttt{1}]$ & Correct on own condition. Error on Entity 2 condition & Complete error information state \\
        F & $[\texttt{0},\texttt{1},\texttt{0},\texttt{1}]$ & Correct on own condition. Error on Entity 2 condition. & Correct on own condition. Error on Entity 1 condition. \\
        $\Psi$ & $[\texttt{1},\texttt{1},\texttt{0},\texttt{1}]$ & Error on own condition. Error on Entity 2 condition. & Correct on own condition. Error on Entity 1 condition. \\
        $\Theta$ & $[\texttt{0},\texttt{1},\texttt{0},\texttt{0}]$ & Correct on own condition. Error on Entity 2 condition & Complete correct information state \\
        E & $[\texttt{1},\texttt{1},\texttt{1},\texttt{1}]$ & Complete error information state & Complete error information state \\
        \bottomrule
    \end{tabular}
    \label{tab:statestable}
\end{table*}

\section{System Overview}\label{sec:systemoverview}
We model the system as a CTMC with $9$ states, determined by $4$ information values. The number of states follows from this remark: if $\widehat{Y}_i=\texttt{1}$, i.e. entity $i\in\{1,2\}$ has incorrect information about itself, the other entity $j=3{-}i$ must also have $\widetilde{Y}_j=\texttt{1}$, i.e. incorrect information about $i$. Thus, states $[\texttt{1},\texttt{*},\texttt{*},\texttt{0}]$ and $[\texttt{*},\texttt{0},\texttt{1},\texttt{*}]$ cannot exist. 
Then, drift can corrupt $\widehat{Y}_i$, which in turn affects $\widetilde{Y}_{j}$, but drift cannot occur directly on $\widetilde{Y}_i$ or $\widetilde{Y}_j$.   

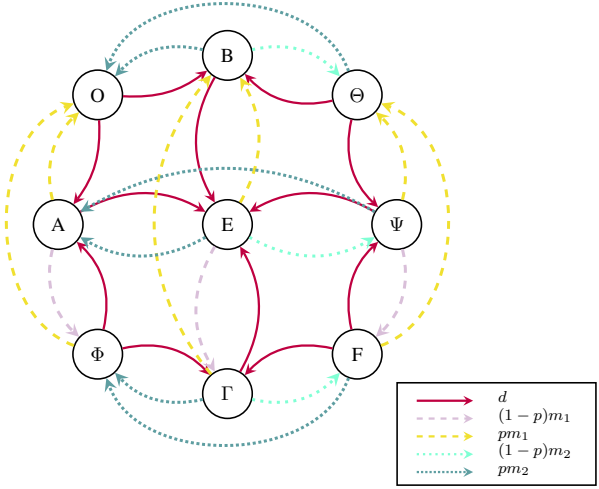
\begin{figure}[t]
    \centering
    \begin{tikzpicture}[->, >=stealth, auto, semithick, scale=0.75, transform shape]
        % Stati esterni
        \node[state] (S1) at (-2.3, 2.3) {O}; 
        % Stati disposti in cerchio
        \node[state] (S2) at (0, 3) {B};
        \node[state] (S3) at (2.3, 2.3) {$\Theta$};
        \node[state] (S4) at (-3, 0) {A};
        \node[state] (S7) at (-2.3, -2.3) {$\Phi$};
        \node[state] (S8) at (0, -3) {$\Gamma$};
        \node[state] (S9) at  (3, 0) {$\Psi$};
+++        \node[state] (S10) at (2.3, -2.3) {F};
        % Stato interno
        \node[state] (S5) at (0, 0) {E};
        
        % Transizioni - Drift (purple)
        \path[purple, line width=0.3mm] (S1) edge[bend right=20] node {} (S2)
                   edge[bend left=20] node {} (S4);
        \path[purple, line width=0.3mm] (S4) edge[bend left] node {} (S5);
        \path[purple, line width=0.3mm] (S10) edge[bend right] node {} (S8);
        \path[purple, line width=0.3mm] (S10) edge[bend left] node {} (S9);
        \path[purple, line width=0.3mm] (S2) edge[bend right] node {} (S5);
        \path[purple, line width=0.3mm] (S7) edge[bend left] node {} (S8);
        \path[purple, line width=0.3mm] (S7) edge[bend right] node {} (S4);
        \path[purple, line width=0.3mm] (S8) edge[bend right] node {} (S5);
        \path[purple, line width=0.3mm] (S3) edge[bend right] node {} (S9);
        \path[purple, line width=0.3mm] (S3) edge[bend left] node {} (S2);
        \path[purple, line width=0.3mm] (S9) edge[bend right] node {} (S5);
        
        % Transizioni - Entity 1 Measurement (Thistle, dashed)
        \path[Thistle, dashed, line width=0.4mm] (S4) edge[bend right] node {} (S7);
        \path[Thistle, dashed, line width=0.4mm] (S5) edge[bend right] node {} (S8);
        \path[Thistle, dashed, line width=0.4mm] (S9) edge[bend left] node {} (S10);
        
        % Transizioni - Entity 1 Measurement and Transmission (Dandelion, dashed)
        \path[Dandelion, dashed, line width=0.4mm] (S7) edge[bend left=70] node {} (S1);
        \path[Dandelion, dashed, line width=0.4mm] (S8) edge[bend left=40] node {} (S2);
        \path[Dandelion, dashed, line width=0.4mm] (S4) edge[bend left] node {} (S1);
        \path[Dandelion, dashed, line width=0.4mm] (S10) edge[bend right=70] node {} (S3);
        \path[Dandelion, dashed, line width=0.4mm] (S9) edge[bend right] node {} (S3);
        \path[Dandelion, dashed, line width=0.4mm] (S5) edge[bend right] node {} (S2);
        
        % Transizioni - Entity 2 Measurement (Aquamarine, dotted)
        \path[Aquamarine, dotted, line width=0.4mm] (S2) edge[bend left] node {} (S3);
        \path[Aquamarine, dotted, line width=0.4mm] (S5) edge[bend right] node {} (S9);
        \path[CadetBlue, densely dotted, line width=0.4mm] (S5) edge[bend left] node {} (S4);
        \path[Aquamarine, dotted, line width=0.4mm] (S8) edge[bend right] node {} (S10);
        
        % Transizioni - Entity 2 Measurement and Transmission (CadetBlue, densely dotted)
        \path[CadetBlue, densely dotted, line width=0.4mm] (S3) edge[bend right=70] node {} (S1);
        \path[CadetBlue, densely dotted, line width=0.4mm] (S2) edge[bend right] node {} (S1);
        \path[CadetBlue, densely dotted, line width=0.4mm] (S9) edge[bend right] node {} (S4);
        \path[CadetBlue, densely dotted, line width=0.4mm] (S8) edge[bend left] node {} (S7);
        \path[CadetBlue, densely dotted, line width=0.4mm] (S10) edge[bend left=70] node {} (S7);

        \node[draw, fill=white, anchor=north west, align=left, font=\footnotesize] at (3,-2.8) {
    \begin{tabular}{l l}
        \begin{tikzpicture} \draw[purple, line width=0.4mm] (0,0) -- (1,0); \end{tikzpicture} & $d$ \\
        \begin{tikzpicture} \draw[Thistle, dashed, line width=0.4mm] (0,0) -- (1,0); \end{tikzpicture} & $(1-p)m_1$ \\
        \begin{tikzpicture} \draw[Dandelion, dashed, line width=0.4mm] (0,0) -- (1,0); \end{tikzpicture} & $p m_1$ \\
        \begin{tikzpicture} \draw[Aquamarine, dotted, line width=0.4mm] (0,0) -- (1,0); \end{tikzpicture} & $(1-p)m_2$ \\
        \begin{tikzpicture} \draw[CadetBlue, densely dotted, line width=0.4mm] (0,0) -- (1,0); \end{tikzpicture} & $p m_2$ \\
    \end{tabular}
};

    \end{tikzpicture}
    \caption{CTMC for BAoII computation. O and $\Phi$ are the two states in which BAoII resets to zero for Entity 1.}
    \label{fig:state_diagram}
\end{figure}

Fig.~\ref{fig:state_diagram} illustrates the full CTMC state space and all possible transitions. The states are defined as described in Tab.~\ref{tab:statestable}.  
We analyze the system from the perspective of entity 1, but, due to its memoryless and symmetric properties, the model equally applies to entity 2.
The system starts in state O, indicating complete knowledge for both entities, or $\Phi$, representing complete knowledge for entity 1 but not for 2. From these states, the system can transition to an error state due to drift, which can corrupt $\widehat{Y}_1$, resulting in A, or $\widehat{Y}_2$, yielding B or $\Gamma$. Subsequently, additional drift may occur, or entities may attempt to improve their knowledge by measuring their own state, transmitting their state to the other entity, or performing both actions simultaneously.  

We define transition matrix $\mathbf{Q}=\{q_{hn}\}_{hn}$, 
%\textcolor{red}{le matrici andrebbero di solito in grassetto. potete controllare di essere stati omogenei?}
where each element $q_{hn}$ represents the transition rate from state $h$ to state $n$ with $h \neq n$. The diagonal elements are defined as $q_{hh} = - \sum_{n \neq h} q_{hn}$.  In a CTMC, the transition rates define the parameters of the exponential distribution that govern the waiting times in each state. 

We introduce some parameters to describe the possible transitions.
Entity $i$ can make a measurement on itself to determine its $\widehat{X}_i$ value, occurring at rate $m$. The transmission of this value to the other entity occurs at a rate $\lambda$, which is proportional to the measurement rate. We set $\lambda = pm$, where $p \in [0,1]$ is a probability that represents the fraction of times the entity transmits its measured value instantaneously.
%In this model, transmissions that are not simultaneous with a measurement are not possible. 
Consequently, with rate $(1 - p)m$, measurements occur without being transmitted.

As $p$ grows, measurement and transmission occur together more frequently, reducing the rate of non-transmitted measurements.
This ensures that an entity obtains correct information about its own state ($\widehat{Y}_i {=} \texttt{0}$) while updating the other entity’s knowledge about its state, making it correct ($\widetilde{Y}_{j} {=} \texttt{0}$).

In Fig.\ \ref{fig:penaltyfunctionbaoii}, we can observe an example of BAoII behavior and how it differs from AoII as the system transitions through various states. After a drift occurs at time $t_1$, both penalty functions begin to increase. At $t_2$, however, a measurement and transmission by entity $1$ take place simultaneously, bringing the BAoII back to a fully correct state. At the same time, the AoII resets to zero due to the update. Then, a drift occurs at $t_3$ and also at $t_4$. Entity $1$ performs a measurement at $t_5$, where the AoII resets while BAoII continues to rise, as it must wait until $t_6$, when both a measurement and a transmission from entity $2$ occur. Only then does the system return to one of the correct information states, $\Phi$ in this case.

With the transition rates defined in the matrix $\textbf{Q}$, we proceed to compute the cycle time $\mathcal{T}$ of the CTMC. We define $\mathcal{T}$ as the expected reset time for the fully correct knowledge states O and $\Phi$, treated as reset states, starting from the first error states. These first error states -— in our case, A, B, and $\Gamma$ —- are reached after a drift causes the system to deviate from the correct knowledge states.

This formulation leverages renewal theory by interpreting each return to a correct knowledge state as the renewal point of a stochastic cycle, enabling the computation of long-term averages over repeated error-recovery processes.
\begin{theorem} \label{th:baoii}
Given the cyclic return time $\mathcal{T}$, the mean long-term BAoII is the average area of the triangle below the BAoII penalty. Thus, it can be expressed as:
\begin{equation}\label{eq:closedbaoii}
    \Delta_{\text{BAoII}} =\frac{\mathcal{T}}{2} = \frac{2p+1}{4pm}
\end{equation}
\end{theorem}

\begin{proof}
    See the Appendix. 
\end{proof}

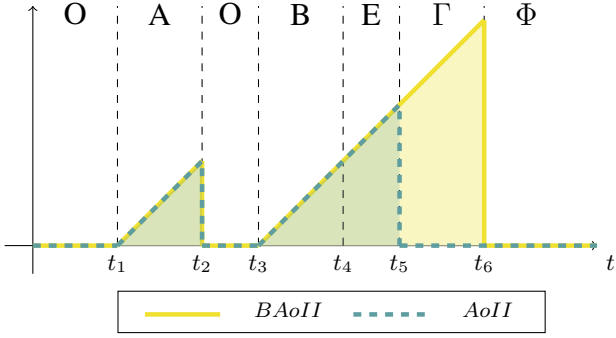
\begin{figure}[t!]
    \centering

\begin{tikzpicture}[scale=0.75]
% Assi
\draw[->] (-0.5,0) -- (10,0); % asse x
\draw[->] (0,-0.5) -- (0, 4.25);  % asse y

% Coordinate dei punti lungo x
\coordinate (d1)    at (1.5,0);
\coordinate (m1t1)  at (3,1.5); % salita pendenza 1
\coordinate (d2)    at (4,0);
\coordinate (m1)    at (8,4);     % salita pendenza 1
\coordinate (m2t2)  at (8,0);

% Triangoli colorati
\fill[Dandelion!30] (d1) -- (m1t1) -- (3,0) -- cycle;
\fill[Dandelion!30] (d2) -- (m1) -- (m2t2) -- cycle;

% Linee tratteggiate verticali
\foreach \x in {1.5,3,4,5.5,6.5,8}
    \draw[dashed] (\x,0) -- (\x,4.25);

% Tratti orizzontali aggiunti
\draw[Dandelion, thick, line width=0.6mm] (0,0) -- (1.5,0);
\draw[Dandelion, thick, line width=0.6mm] (3,0) -- (4,0);
\draw[Dandelion, thick, line width=0.6mm] (8,0) -- (10,0);

% Tratti obliqui e verticali
\draw[Dandelion, thick, line width=0.6mm] (d1) -- (m1t1);
\draw[Dandelion, thick, line width=0.6mm] (m1t1) -- (3,0);

\draw[Dandelion, thick, line width=0.6mm] (d2) -- (m1);
\draw[Dandelion, thick, line width=0.6mm] (m1) -- (m2t2);

% Tratti orizzontali aggiunti
\draw[CadetBlue, dashed, line width=0.6mm] (0,0) -- (1.5,0);
\draw[CadetBlue,dashed, line width=0.6mm] (3,0) -- (4,0);
\draw[CadetBlue,dashed, line width=0.6mm] (6.5,0) -- (10,0);

\draw[CadetBlue, dashed, line width=0.6mm] (1.5,0) -- (3,1.5) -- (3,0);
\draw[CadetBlue, dashed, line width=0.6mm] (4,0) -- (6.5,2.5) -- (6.5,0);
%\draw[CadetBlue, dashed, line width=0.6mm] (6.5,0) -- (8,1.5) -- (8,0) -- cycle;
\fill[CadetBlue, fill opacity=0.2] (1.5,0) -- (3,1.5) -- (3,0) -- cycle;
\fill[CadetBlue, fill opacity=0.2] (4,0) -- (6.5,2.5) -- (6.5,0) -- cycle;
%\fill[CadetBlue, fill opacity=0.2] (6.5,0) -- (8,1.5) -- (8,0) -- cycle;

% Etichette in basso
\node[below] at (1.5,0)   {\small $t_1$};
\node[below] at (3,0) {\small $t_2$};
\node[below] at (4,0)   {\small $t_3$};
\node[below] at (5.5,0) {\small $t_4$};
\node[below] at (6.5,0) {\small $t_5$};
\node[below] at (8,0) {\small $t_6$};

% Etichette sopra (centrate tra le linee tratteggiate)
\foreach \x/\l in {0.75/O, 2.25/A, 3.5/O, 4.75/B, 6/E, 7.25/$\Gamma$, 8.75/$\Phi$}
    \node[above] at (\x,3.75) {\large \l};

% Etichetta asse x (t di tempo)
\node[below right] at (10,0) {\small $t$};

\node[draw, fill=white, anchor=north west, align=left, font=\footnotesize] at (1.5,-0.8) {
\begin{tabular}{l l l l}

    \begin{tikzpicture} \draw[Dandelion, thick, line width=0.6mm] (0,0) -- (1,0); \end{tikzpicture} & $BAoII$ \quad
    \begin{tikzpicture} \draw[CadetBlue, dashed, line width=0.6mm] (0,0) -- (1,0); \end{tikzpicture} & $AoII$ \\
\end{tabular}

};

% Etichetta asse y (BAOII)
%\node[above left, rotate=90] at (0,4.25) {\small $\Delta_{BAOII}$};
\end{tikzpicture}
\caption{Comparison of BAoII and AoII in the bi-entity system following a sequence of state transitions in the CTMC. Markers $t_1 - t_6$ indicate key events such as drifts, measurements, and transmissions.}
    \label{fig:penaltyfunctionbaoii}
\end{figure}

Each transition has an associated cost: $k_m$ for a measurement and $k_{\lambda}$ for a transmission. If both actions occur within the same transition, the total cost is $k_m + k_{\lambda}$. However, a holding cost is associated with each state, corresponding to measurement and transmission actions performed without initiating a state transition.
Thus, we define the average long-term cost for an entity as
\begin{equation}
    K = k_m\cdot m + k_{\lambda} \cdot pm
\end{equation}
which accounts for the cost of individual measurements and measurements immediately followed by a transmission.
Although BAoII should ideally remain close to zero, this would require frequent measurements and high transmission rates, which would lead to significant costs. Thus, a trade-off between cost and information accuracy must be established.
The total cost associated with BAoII, including both the penalty for incorrect information and the total cost of actions, is
\begin{equation}\label{eq:kbaoii}
   \mathcal{C} = \Delta_{\text{BAoII}} + K
\end{equation}
where the first term represents the BAoII penalty, which increases with the time needed to restore full system knowledge. Minimizing BAoII requires frequent measurements and transmissions, leading to higher costs. Thus, the problem reduces to optimizing $\mathcal{C}$, that is,
\begin{equation}\label{eq:totalcost}
    \mathcal{C} = \frac{2p+1}{4pm} + (k_m + pk_{\lambda}) m
\end{equation}

Given the trade-off between transmission costs and the BAoII penalty, it is essential to investigate optimal transmission policies under different cost constraints. Our goal is to identify a cost threshold below which the optimal strategy is to always transmit immediately after each measurement. 
From this observation, we get this theorem.

\begin{theorem}
There exists a threshold value
$k^\circ_{\lambda}$ such that, for any transmission cost $k_{\lambda} < k^\circ_{\lambda} $, the optimal transmission policy minimizing the total cost $\mathcal{C}$ is to always transmit after every measurement, i.e., $p^* = 1$. This threshold is found as
\begin{equation}\label{eq:kt}
    k^\circ_{\lambda} = \frac{1}{4m^2}
\end{equation}
\end{theorem}

\begin{proof}
To determine whether \( p = 1 \) minimizes the total cost \( \mathcal{C} \), we take the partial derivative of \( \mathcal{C} \) with respect to \( p \), set it equal to zero and evaluate the condition under which \( p = 1 \) satisfies the optimality criterion. Solving the resulting expression for \( k_{\lambda} \) yields the threshold value \( k^\circ_{\lambda} \). To confirm that this corresponds to a minimum, we evaluate the second derivative of \( \mathcal{C} \) with respect to \( p \) and observe that it is strictly positive for all \( p > 0 \) and \( m > 0 \), ensuring that \( p = 1 \) is a local minimum.
\end{proof}

\section{Results}\label{sec:results}
We present numerical results based on the mathematical analysis of the CTMC model. Furthermore, we demonstrate the applicability of BAoII for evaluating the performance of real-world VDE scenarios.

From (\ref{eq:closedbaoii}), we observe that, after algebraic simplification, the value of $\Delta_{\text{BAoII}}$ is independent of the drift $d$ and depends only on the measurement rate $m$ and the probability $p$.
Consequently, in Fig.~\ref{fig:cycle duration}  we show how $\Delta_{\text{BAoII}}$ varies with the rate $m$ for different values of $p$. 
To minimize $\Delta_{\text{BAoII}}$, the optimal strategy is to increase $m$ and set $p=1$. This guarantees that all measurements are promptly transmitted to the other entity, enabling the system to quickly attain a fully correct state.

%As previously observed in (\ref{eq:kbaoii}), the goal of minimizing $\Delta_{\text{BAoII}}$ is subject to a trade-off influenced by measurement and transmission costs.

However, as shown in (\ref{eq:kbaoii}), minimizing $\Delta_{\text{BAoII}}$ must be balanced against the associated measurement and transmission costs.
Thus, we consider a high-immersion VDE in which the costs are modeled as proportional to the packet size (in Mbytes), with time measured in seconds and the measurement rate in Hz.
We assume that the packets required to update the user state originate from controller inputs, which facilitate movement in the virtual environment, or from haptic feedback. These packets are assumed to have a maximum size of $500$ bytes. For transmission, we account for packets that also carry user position tracking data, which are necessary for visual updates and typically have an average size of $1$ kbytes. 

Fig. \ref{fig:costapp1} shows the total cost $\mathcal{C}$ of this scenario. We can see that $p=1$ minimizes $\mathcal{C}$ for rates up to approximately $20$ Hz, corresponding to an average latency of $50$ ms. However, as the rate increases, $p=1$ becomes less efficient due to the larger transmission cost. In such cases, tolerating a higher BAoII becomes preferable, as it allows updates to be transmitted less frequently, reducing transmission costs.

\begin{figure}[t]
    \centering
    \includegraphics[width=\linewidth]{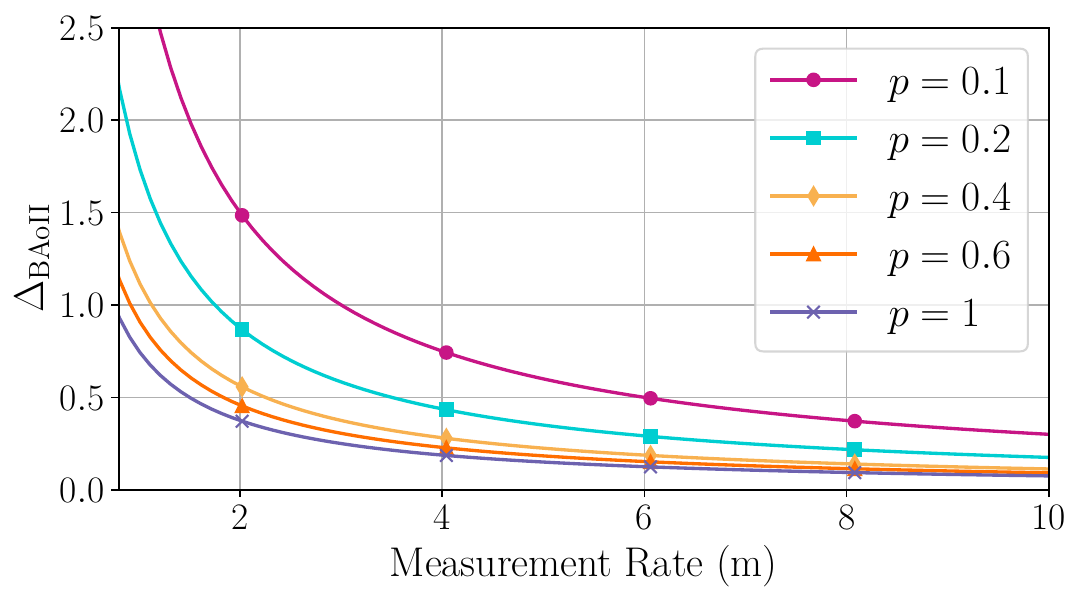}
    \caption{Expected $\Delta_{\text{BAoII}}$ as a function of the measurement rate $m$, for different values of $p$.}
    \label{fig:cycle duration}
\end{figure}
\begin{figure}[tbp]
    \centering
    \includegraphics[width=\linewidth]{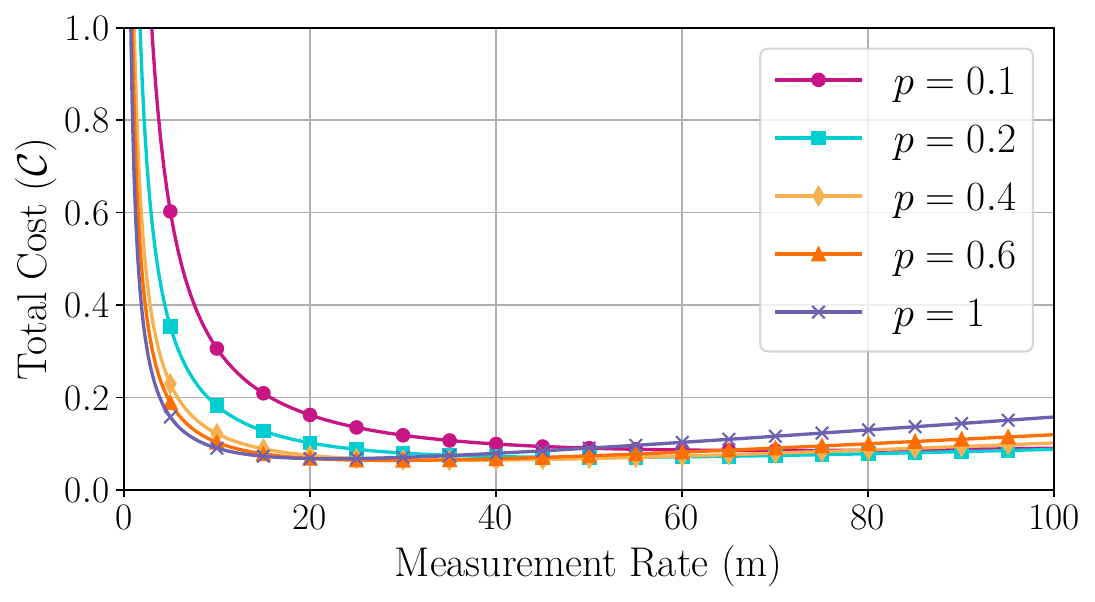}
    \caption{Total cost $\mathcal{C}$ vs $m$ with various values of $p$ and with $k_m =500$B and $k_{\lambda}=1$kbytes }
    \label{fig:costapp1}
\end{figure}
\begin{figure}[t]
    \centering
    \includegraphics[width=\linewidth]{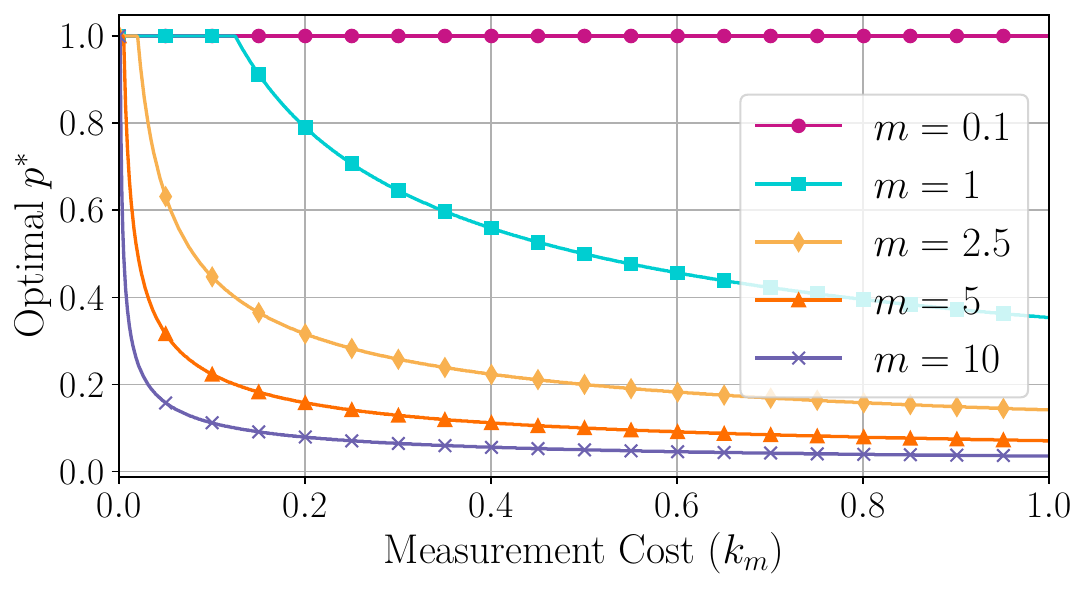}
    \caption{Optimal value of $p$ versus imposed measurement rate $0\leq k_m \leq 1$ Mbytes and $k_{\lambda} = 2k_m$}
    \label{fig:optimalp}
\end{figure}

Next, we assume that the transmission cost and the measurement cost are linearly related \cite{van2024impact}, $k_{\lambda} = \eta \cdot k_m$, where $\eta$ represents the cost ratio, we can rewrite (\ref{eq:kt}) as follows:
\begin{equation}
    k^\circ_m = \frac{1}{4\eta m^2}.
\end{equation}
%We proceed to illustrate how both transmission costs and rates affect the optimal value of $p$. 
Thus, we analyze how transmission costs and measurement rates affect the optimal value of $p$. We take $k_{\lambda} = 2k_m$ \cite{plopski2022eye} and compute the partial derivative of $\Delta_{\text{BAoII}}$ with respect to $p$. We then analyze the resulting optimal value $p^*$ in different configurations of cost and transmission rate, as shown in Fig.~\ref{fig:optimalp}. As transmission costs $k_{\lambda}$ and measurement rate $m$ increase, the cost interval over which $p = 1$ remains optimal becomes progressively narrower. In particular, for $m = 0.1$, the policy with $p = 1$ remains optimal even with high transmission costs. This is due to the relatively low number of transmissions and measurements that occur in a unit time. However, this behavior holds only within a certain cost range; for significantly higher costs, even a low $m$ would no longer sustain $p = 1$ as the optimal policy. These results highlight the importance of evaluating both the cost constraints and the desired transmission rate when designing systems for different application scenarios.
\begin{figure}[t!]
    \centering
    \includegraphics[width=\linewidth]{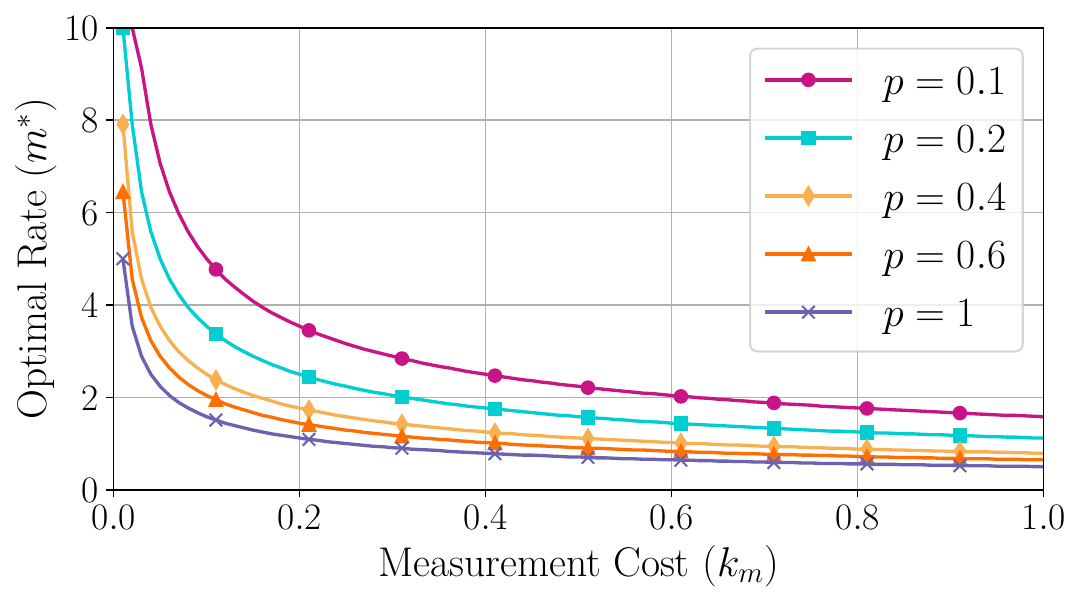}
    \caption{Optimal value of measurement rate versus imposed measurement rate $0\leq k_m \leq 1$ Mbytes and $k_{\lambda} = 2k_m$} 
    \label{fig:optimalm}
\end{figure}

Conversely, Fig.~\ref{fig:optimalm} illustrates how the optimal measurement rate $m$ varies with cost for different values of $p$. For $p = 1$, the optimal values of $m$ are the lowest among all configurations, while for $p = 0.1$, significantly higher values of $m$ become feasible. However, frequent measurements coupled with infrequent transmissions do not effectively minimize $\Delta_{\text{BAoII}}$, see also Fig.~\ref{fig:cycle duration}. Nonetheless, this strategy allows systems to tolerate higher costs, albeit with an increased $\Delta_{\text{BAoII}}$.

Now, consider a generic $\eta$ to study the maximum $k_m$ that can be imposed in a given scenario to achieve a desired measurement rate. In contrast, we can determine the maximum achievable measurement rate given a specific cost model.

In Fig.~\ref{fig:measurement_cost}, %shown on a logarithmic scale for clarity, 
we illustrate the maximum measurement cost that can be tolerated to support a given measurement rate $m$ under various values of $\eta$.
To contextualize the impact of these constraints, we analyze two distinct types of VDEs with markedly different requirements, as described in Sec.\ \ref{sec:applications}.
    
The first scenario involves highly immersive collaboration or real-time online interaction \cite{favero2024strategic, dohler2025crucial}, where updates are required within the range of 10 to 100 milliseconds \cite{8933555}, corresponding to measurement rates $ 10 \text{Hz}\leq m\leq100 \text{Hz}$. In this case, the maximum allowable costs are considerably lower as frequent transmissions are essential. These systems typically rely on less reliable but faster communication protocols. The acceptable cost range in such contexts varies from a maximum of $10^1$ to a minimum of $10^{-7}$.

In the second scenario, we consider communicating DTs for logistics\cite{favero2024logistics}, where updates are typically required every 10 to 100 seconds. These applications generally operate over more reliable communication protocols with higher transmission costs \cite{allam2022metaverse}. Depending on the application-specific value of $\eta$, the maximum allowable measurement cost ranges from $10^5$ to $10^{-1}$, reflecting scenarios where transmission is significantly more expensive than measurement.

\begin{figure}
    \centering
    \includegraphics[width=\linewidth]{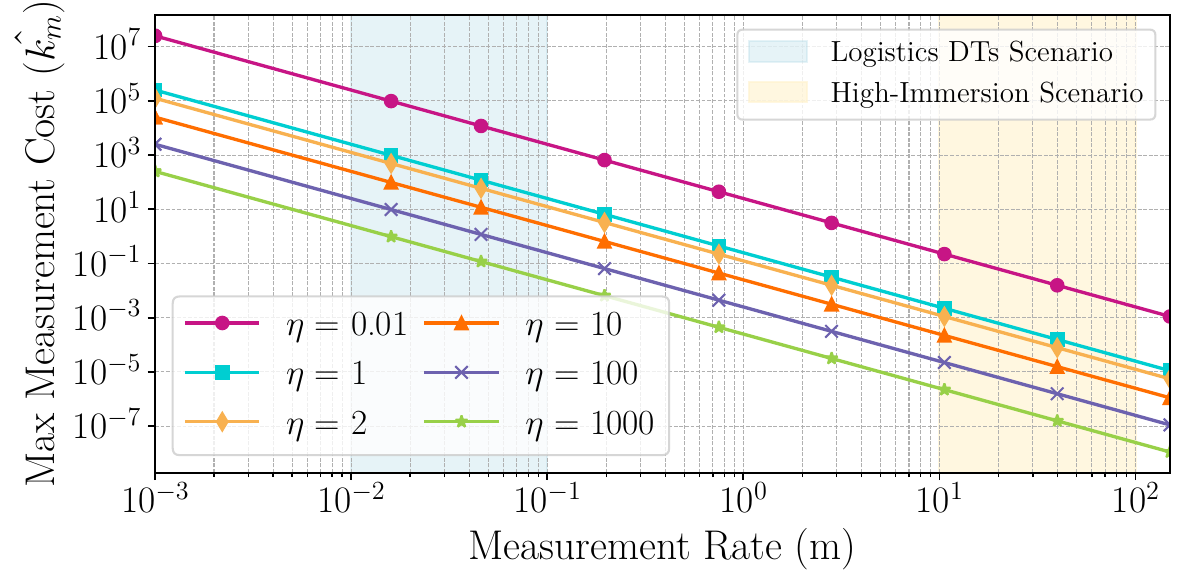}
    \caption{Maximum allowable measurement cost for a given measurement rate $m$, under different values of $\eta$.}
    \label{fig:measurement_cost}
\end{figure}

\section{Conclusions}\label{sec:conclusions}
We introduced BAoII, a metric related to age of incorrect information that extends status freshness and accuracy to bidirectional exchanges, to assess mutual awareness of agents in dynamic systems like the Metaverse and VR. 
Through a CTMC analysis, we demonstrated that minimizing BAoII requires a careful balance between the measurement frequency $m$ and the probability $p$ of transmitting the measurement. We proved that a persistent transmission policy ($p=1$) minimizes BAoII and we found the cost threshold under which this strategy minimizes the total cost $\mathcal{C}$. This threshold can be used in various VDEs to bridge system cost models with performance requirements. Specifically, it enables either the estimation of achievable measurement rates given cost constraints or, conversely, the derivation of cost settings required to sustain a desired measurement rate. This makes our analysis valuable for adaptive system design under resource limitations.

BAoII can be a valuable tool for real-time collaboration, with future applications including energy harvesting sensor systems \cite{hatami2021aoi} and semantic communications \cite{8930830}. %9838737}.
Future works may also involve extending the analysis to more complex multi-entity systems for large-scale VDE, addressing cybersecurity challenges in scenarios with untrusted entities\cite{bonagura2025strategic}, and employing game-theoretic approaches to evaluate the impact of entities selfish behavior on the system.

\bibliographystyle{IEEEtran}
\bibliography{IEEEabrv, biblio}

\appendix \label{sec:appendix}
This is a detailed derivation of the calculations that led to the formulation of $\Delta_{\text{BAoII}}$ in Theorem \ref{th:baoii}.

We start by assuming that the CTMC is ergodic, which allows us to compute the probability that entity 1 is in one of the two fully accurate information states, namely O and $\Phi$.
Let $\mathcal{S} = \{\text{O}, \Phi, \text{A}, \text{B}, \Gamma, \text{E}, \text{F}, \Psi, \Theta \}$ be the set of all states, and let $\textbf{Q} \in \mathbb{R}^{9\times9}$ be the transition matrix, where each entry represents a transition rate of the CTMC depicted in Fig.~\ref{fig:state_diagram}. 
%The steady-state condition for all states is given by $\boldsymbol{\pi} \mathbf{Q} = \mathbf{0}$ under $\sum_{h \in \mathcal{S}} \pi_h = 1$.
The steady-state condition for all states is given by $\boldsymbol{\pi} \mathbf{Q} = \mathbf{0}$, where $\boldsymbol{\pi} = (\pi_i)$ with $i \in \mathcal{S}$ is the stationary distribution vector and $\mathbf{0}$ denotes a zero vector of size 9. The normalization condition is $\sum_{s \in \mathcal{S}} \pi_s = 1$.
%\textcolor{red}{Sto vettore $\pi$ e' mai definito? Bisogna anche scrivere che lo 0 dell'equazione (che vi ho messo io in grassetto) e' un vettore tutto nullo di 9 elementi.}

For each state $n$, the sum of the transition rates weighted by their respective probabilities must be zero:
\begin{equation}\label{eq:pi2}
%\sum_{n=1}^{9} \pi_n q_{nh} = 0, \quad \forall h \in \{1, \dots, 9\}.
\sum_{h \in \mathcal{S} } \pi_n q_{nh} = 0 \quad \forall n \in \mathcal{S}.
\end{equation}
Solving this system yields the steady-state probabilities $\boldsymbol{\pi}.$ %= (\pi_{\text{O}}, \pi_{\Phi}, \dots, \pi_{\Theta})$.
%\textcolor{red}{non si chiamano mica 1,.. 9}
Among them, in particular, we need the following steady-state probabilities. %\textcolor{red}{controllate che mi sa che quei p-1 scritti sotto sono degli 1-p. forse l'avete ricavato con wolfram ma insomma mettiamo le probabilita' come tali...}
\begin{IEEEeqnarray}{rCl}
\small
\pi_{\Phi} &=& \frac{d m_1 m_2 (1 {-} p)^2 p}{(d {+} m_1 (1 {-} p) p)(d {+} m_2 (1 {-} p) p)(d {+} m_1 {-} m_1 p^2)} \quad \\
\pi_{\text{O}} &=& \frac{m_1 m_2 (1 {-} p)^2 p^2}{(d {+} m_1 (1 {-} p) p)(d {+} m_2 (1 {-} p) p)} \, .
\end{IEEEeqnarray}
We define the conditioned probabilities of starting in states O and $\Phi$ as $P_{\text{O}}$ and $P_{\Phi}$, respectively:
\begin{equation}
    P_{\text{O}} = \frac{\pi_{\text{O}}}{\pi_{\text{O}} + \pi_\Phi} \qquad P_\Phi = \frac{\pi_\Phi}{\pi_{\text{O}} + \pi_\Phi} \; .
\end{equation}
To compute the time spent in erroneous states, we treat O and $\Phi$ as reset states. 
%\textcolor{red}{non e' vero. e' che $\{ 0, \Phi\}$ e' un absorbing set}
Then, we determine the expected time required to enter the reset state.
%\textcolor{red}{qua serve essere piu' rigorosi. Se diciamo che e' assorbente, non e' che ci usciamo e poi ci ritorniamo. Forse manco conviene chiamarlo assorbente}

Transitions away from O and $\Phi$ into incomplete correctness states occur only through drift. Specifically:
\begin{itemize}
    \item From O, drift can lead to states A or B.
    \item From $\Phi$, drift can result in transitions to A or $\Gamma$.
\end{itemize}
Thus, A, B, and $\Gamma$ are the first erroneous states, representing the first states of incorrect information that can be reached from O and $\Phi$. 
We aim to compute the transition probabilities to these states and, subsequently, determine the system's expected reset time from each of them.

We define conditional transition probabilities $s_{h,n}$, representing the probability of transitioning to state $h$ given that the process starts in state $n$, as:
\begin{equation}\label{eq:transizionicondizionate}
   s_{h,n}= \frac{q_{hn}}{\sum_{l \neq 0}q_{hl}}
\end{equation}  
where $l$ represents the other transition rates present in the corresponding row of the transition rate matrix $\textbf{Q}$ for the initial state considered.
Next, we define the system's average reset time, representing the average duration of a CTMC cycle, as the expected time to return to either O or $\Phi$ after reaching one of the first three erroneous states A, B, or $\Gamma$:

\begin{equation}\label{eq:tass}
    \mathcal{T} = P_{\text{O}}(s_{\text{O},\text{A}}\cdot\tau_{\text{A}} + s_{\text{O},\text{B}}\cdot\tau_{\text{B}}) + P_{\Phi}(s_{\Phi,\text{A}}\cdot\tau_{\text{A}} + s_{\Phi,\Gamma}\cdot\tau_{\Gamma})
\end{equation}

where $\tau$ represents the average reset time starting from a given state.
%\subsection{Absorption Times and Conditional Transitions}
 The reset times $\tau_{\text{A}}$, $\tau_{\text{B}}$, and $\tau_{\Gamma}$ are calculated using a system of equations derived from the conditional transition probabilities described in (\ref{eq:transizionicondizionate}). The system of equations for the reset times is as follows:
%Below, we present now the system of equations used to compute the absorption times $\tau_{\alpha}$, $\tau_{\beta}$, and $\tau_{\Gamma}$, as given in (\ref{eq:tass}). These equations are formulated using the conditional transition probabilities described in (\ref{eq:transizionicondizionate}).
%\textcolor{red}{dovremmo fare uno sforzo per levare 'sto small}
\begin{IEEEeqnarray}{c}
\small
\left\{ \,
\begin{IEEEeqnarraybox}[][c]{l?s}
  \IEEEstrut
  \tau_{\text{A}} = \frac{2}{d+m_1} + \frac{d}{d+m_1} \tau_{\text{E}} & \\
    \tau_{\Gamma} = \frac{1}{m_2+d+p m_1} 
  + \frac{(1-p) m_2}{m_2+d+p m_1} \tau_{\text{F}} & \\
  \quad + \frac{d}{m_2+d+p m_1} \tau_{\text{E}} 
  + \frac{p m_1}{m_2+d+p m_1} \tau_{\text{B}} & \\
  \tau_{\text{E}} = \frac{p m_1}{m_1+m_2} \tau_{\text{B}} 
  + \frac{(1-p) m_1}{m_1+m_2} \tau_{\Gamma} & \\
   \quad + \frac{p m_2}{m_1+m_2} \tau_{\text{A}} 
  + \frac{(1-p) m_2}{m_1+m_2} \tau_{\Psi} & \\
   \tau_{\text{F}} = \frac{p m_1}{p m_1+p m_2+2d} \tau_{\Theta} 
  + \frac{1}{p m_1+p m_2+2d} & \\
   \quad + \frac{d}{p m_1+p m_2+2d} \tau_{\Gamma} 
  + \frac{d}{p m_1+p m_2+2d} \tau_{\Psi} & \\
   \tau_{\Psi} = \frac{p m_1}{m_1+d+p m_2} \tau_{\Theta} 
  + \frac{(1-p) m_1}{p m_1+p m_2+2d} \tau_{\text{F}} & \\
   \quad + \frac{d}{p m_1+p m_2+2d} \tau_{\text{E}} 
  + \frac{p m_2}{p m_1+p m_2+2d} \tau_{\text{A}} & \\
   \tau_{\text{B}} = \frac{1}{m_2+d} 
  + \frac{(1-p) m_2}{m_2+d} \tau_{\Theta} 
  + \frac{d}{m_2+d} \tau_{\text{E}} & \\
  \tau_{\Theta} = \frac{1}{p m_2+2d} 
  + \frac{d}{p m_2+2d} \tau_{\text{B}} 
  + \frac{d}{p m_2+2d} \tau_{\Psi} &
  \IEEEstrut
\end{IEEEeqnarraybox}
\right.
\label{eq:system_tau}
\end{IEEEeqnarray}

This system of equations allows us to compute the individual reset times for each state. By solving this system, we obtain the following expressions for the reset times:
\begin{IEEEeqnarray}{rCl}
\small
    \tau_{\Theta} &= \frac{2p d m_2 + p m_1 m_2 + d m_1 + m_1^2}
    {p m_1 m_2 (p m_2 + 2d + m_1)} \\
    \tau_{\text{F}} &= \frac{2p d m_2 + p m_1 m_2 + d m_1 + m_1^2}
    {p m_1 m_2 (p m_2 + 2d + m_1)} \\
    \tau_{\text{B}} &= \frac{2p d m_2 + p m_1 m_2 + d m_1 + m_1^2}
    {p m_1 m_2 (p m_2 + 2d + m_1)} \\
    \tau_{\Gamma} &= \frac{2p d m_2 + p m_1 m_2 + d m_1 + m_1^2}
    {p m_1 m_2 (p m_2 + 2d + m_1)} \\
    \tau_{\text{A}} &= \frac{2p^2 m_2^2 + 2p d m_2 + 2p m_1 m_2 + d m_1}
    {p m_1 m_2 (p m_2 + 2d + m_1)} \\
    \tau_{\Psi} &= \frac{2p^2 m_2^2 + 2p d m_2 + d m_1 + m_1^2}
    {p m_1 m_2 (p m_2 + 2d + m_1)} \\
    \tau_{E} &= \frac{2p^2 m_2^2 + 2p d m_2 + d m_1 + m_1^2}
    {p m_1 m_2 (p m_2 + 2d + m_1)} \, .
\IEEEyesnumber
\end{IEEEeqnarray}

%Substituting these results into (\ref{eq:tass}) allows us to first compute $\mathcal{T}$ and then $\Delta_{\text{BAoII}}$, ultimately deriving its closed-form expression as given in (\ref{eq:closedbaoii}).
Substituting these results for the individual reset times into (\ref{eq:tass}), we obtain:
\begin{equation}
\displaystyle
     \mathcal{T}   = \frac{1}{m_1} + \frac{1}{2 p m_2} \, .
\end{equation}
In particular, if $m_1 = m_2$, we obtain:
\begin{equation}
\displaystyle
   \mathcal{T}   = \frac{1}{m} + \frac{1}{2pm} \, .
\end{equation}

\balance
\end{document}